\documentclass[njp]{revtex4}

\usepackage[pdfversion=1.4]{hyperref}
\usepackage{graphicx,color,amssymb,amsthm}

\usepackage{dcolumn}
\usepackage{bm}
\usepackage{bbm}
\usepackage{amsmath}
\usepackage{soul}
\usepackage{enumerate}

\newcommand{\braket}[2]{\langle #1|#2\rangle}

\newcommand{\bra}[1]{\langle {#1} |}
\newcommand{\ket}[1]{| {#1} \rangle}

\newcommand{\wt}[1]{\widetilde{#1}}

\usepackage{multirow,array,rotating}
\newcolumntype{M}{>{$\vcenter\bgroup\hbox\bgroup}c<{\egroup\egroup$}}
\newcommand{\ve}{\boldsymbol}

\def\>{\rangle}
\def\<{\langle}

\def\I{ \mathbbm{1} }

\DeclareMathOperator{\tr}{tr}

\newtheorem{lem}{Lemma}

\definecolor{nblue}{rgb}{0.3,0.3,1.0}
\definecolor{ngreen}{rgb}{0.2,0.7,0.2}
\definecolor{nred}{rgb}{0.9,0.1,0}
\definecolor{norange}{rgb}{0.8,0.5,0}

\def\>{\rangle}
\def\<{\langle}

\def\I{ \mathbbm{1} }

\def\dim{ \mathrm{dim}}

\usepackage[normalem]{ulem}  
\renewcommand\sout{\bgroup \color{red} \ULdepth=-.5ex \ULset}

\begin{document}

\title{{Frustration-free Hamiltonians supporting Majorana zero edge modes}}

\author{Sania Jevtic}\email{sania.jevtic@imperial.ac.uk}
\author{Ryan Barnett}
\affiliation{Department of Mathematics, Huxley Building, Imperial College, London SW7 2AZ, United Kingdom}

\begin{abstract}
A one-dimensional fermionic system, such as a superconducting wire, may host Majorana zero-energy edge modes (MZMs) at its edges when it is in the topological phase. MZMs provide a path to realising fault-tolerant quantum computation, and so are the focus of intense experimental and theoretical studies. However, given a Hamiltonian, determining whether MZMs exist is a daunting task as it relies on knowing the spectral properties of the Hamiltonian in the thermodynamic limit.
The Kitaev chain is a paradigmatic non-interacting model that supports MZMs and the Hamiltonian can be fully diagonalised.  However, for interacting models, the situation is far more complex.
Here we consider a different classification of models, namely, ones with frustration-free Hamiltonians. Within this class of models, interacting and non-interacting systems are treated on an equal footing, and we identify exactly which Hamiltonians can realise MZMs.
\end{abstract}

\date{\today}

\maketitle

\section{Introduction}

\begin{figure}[b]
\includegraphics[trim=0cm 2cm 0cm 3cm, clip,width=6in]{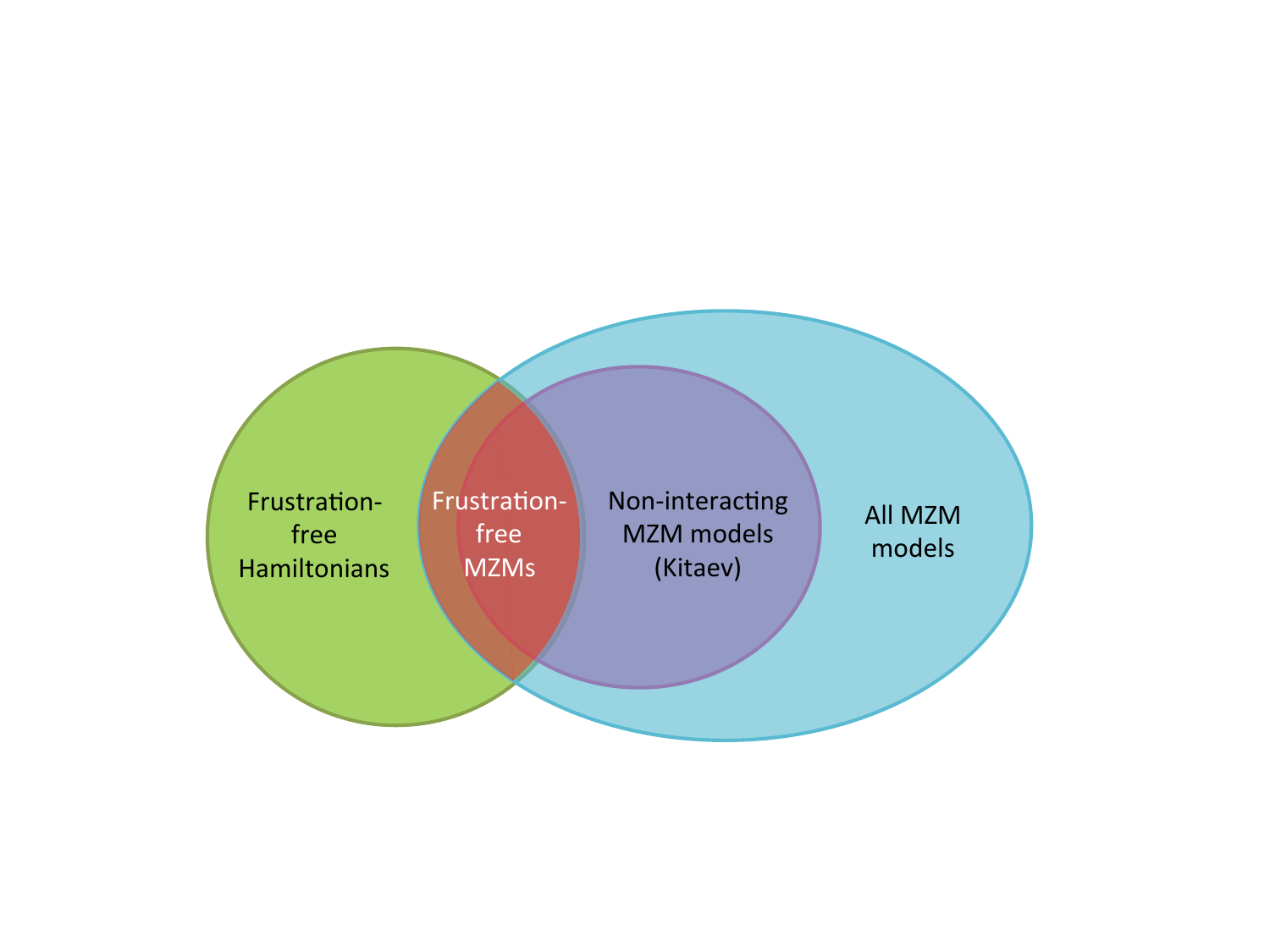}
\caption{Our goal is to find the set of frustration-free Majorana zero mode (MZM) models (the red region). We shall see that this set includes both interacting and non-interacting MZM models.}
\label{FF_venn}
\end{figure}

Majorana fermions were first conceived by Ettore Majorana as a real solution to the Dirac equation \cite{E_Majorana_1937}, and being real means that these fermions are their own anti-particles. In condensed matter systems, the Majorana fermion is a Hermitian fermionic operator. It is of particular importance when it commutes with the system's Hamiltonian and corresponds to a localised zero-energy mode. The reason for this is that the presence of such modes indicates that the fermionic system is in a topologically-ordered phase. The prototypical example of such a system is a theoretical model for a superconducting wire, the ``Kitaev chain'' \cite{Kitaev}, and it describes a line of non-interacting spinless (complex) fermions, where the Cooper pairs are bound via a $p$-wave pairing. Kitaev demonstrated that, for certain ranges of the Hamiltonian parameters, there is a phase where two gapless boundary modes exist and these correspond to Majorana fermions exponentially localised at opposite end of the chain. The modes persist even in the presence of symmetry-preserving local perturbations. We will call such topologically-protected Majorana fermions ``Majorana zero modes'' (MZMs).

The presence of two MZMs implies the Hamiltonian possesses a two-fold degenerate ground space, therefore this degeneracy is also robust against local noise and the ground states cannot be distinguished by local measurements. The ground space states can be used to encode quantum information in a fault-tolerant way.
Subsequent studies have demonstrated that, indeed, systems with MZMs form promising candidates for building a quantum memory \cite{Q_Memory}.  Furthermore, as MZMs display non-Abelian statistics upon braiding, a collection of Majorana wires could provide a realisation of a topological quantum computer. For a review of this topic see Refs.~\cite{alicea2012, SFN_2015}. 
A number of theoretical proposals,  including those, for instance, in Refs.~\cite{oreg2010,lutchyn2010,alicea2011,jiang2011}, have addressed the challenging task of
producing and controlling MZMs in systems ranging from cold atomic gases to solid state materials.
Signatures of MZMs have been reported in several recent experiments at the interface of topological insulators and superconductors \cite{Expt1_2012,Expt2_2012,Expt3_2012,Expt1_2013}, in ferromagnetic chains on a superconductor \cite{Expt_2014}, and in quantum spin liquids \cite{Expt_2016}.

The Kitaev model, though the subject of substantial research effort, is somewhat idealised in that it neglects interactions between the fermionic particles.  
The robustness of topological phases with respect to interactions is an important question, especially for practical purposes. In fact,  
in some cases, repulsive interactions may be beneficial as they have been shown to lead to a stabilisation of the topological phase~\cite{SASF_2011,SAR_2011,Hassler_Schuricht,Thomale_etal}.

In this manuscript, we analyse a particular class of interacting and non-interacting 1D Hamiltonians to determine whether MZMs are present. The class that we are interested in is the collection of Hamiltonians with two-body and nearest-neighbour interactions that are \emph{frustration-free}. Let $L$ be the number of spins  in the system, then such a Hamiltonian is given by $H=\sum_{i=1}^{L-1} h_{i,i+1}$, and $H$ being frustration-free (FF) means that the ground state of $H$ is simultaneously a ground state of each dimer term $h_{i,i+1}$ for all $i = 1,\dots,L$. This is a non-trivial statement as the $h_{i,i+1}$ may not commute with one another, and eigenstates that are not the ground states may be frustrated. 
Frustration free (FF) systems form an important class of many-body local Hamiltonians. Examples of FF spin systems include the toric code \cite{Kitaev_toric} (actually the toric Hamiltonian is a sum of commuting terms) and the AKLT model \cite{AKLT}. They are also used in the study of tensor network states: given a matrix product state, it is possible to construct a parent Hamiltonian that is gapped and frustration free \cite{PEPS_inj}.

Frustration-free fermionic Hamiltonians have been studied in Ref.\ \cite{Katsura_etal} for the purposes of identifying phases with MZMs. In their work, they begin with a physical fermion Hamiltonian with the appropriate symmetries then demand that it be unfrustrated. They indeed find interacting Hamiltonians in topological phases supporting MZMs. Our work complements these results by approaching the task from a more abstract setting which allows a complete systematic deduction of all FF Hamiltonians, and we can conclusively say there are no additional FF phases for MZMs. 
A central ingredient in our work is Ref.\ \cite{Bravyi_Gosset} which fully characterises the ground space and spectral properties of all 1D FF spin Hamiltonians with two-body and nearest-neighbour interactions. We apply a Jordan-Wigner transformation to map the spin Hamiltonians to fermionic ones. The Jordan-Wigner transformation has proven to be a very useful tool for analyzing fermionic systems, especially those in one spatial dimension, and the connection between spin and fermion pictures in the Kitaev chain has been thoroughly explored  and extended to parafermions, the higher-dimensional analogues of MZMs, in Ref.~\cite{Fendley}.
By virtue of the Jordan-Wigner transformation, the spectral and frustration free properties carry over to the fermion picture, however, as the transformation is non-local, we are not guaranteed a local fermionic Hamiltonian. Nevertheless, a Hamiltonian for Majoranas must preserve fermionic parity (fermion number modulo 2), and imposing this symmetry results in local fermionic Hamiltonians. In this way we obtain a complete characterisation of all gapped FF fermionic Hamiltonians. From this, we may then identify the topologically-ordered systems. We find that they are Kitaev chains, which may be non-interacting or interacting, with nearest-neighbour interactions that are either attractive or repulsive. Figure \ref{FF_venn} helps visualise the set of Hamiltonians we are interested in.

A noteworthy observation that arises from analysing the FF spin Hamiltonians is a connection between them and the ``one-dimensional line" (ODL) of Peschel and Emery \cite{Peschel_Emery}. The ODL is a line in the phase diagram where the time evolution operator of a kinetic spin model can be related to the transfer matrix of certain Ising models. An example of an ODL occurs in the axial nearest-neighbour Ising (ANNNI) model \cite{Peschel_Emery,Baccaria_2006}. Along the ODL, the ANNNI model is dual to the FF Hamiltonian, Eq.\eqref{eq_H'}, that we find is capable of hosting MZMs in the fermion picture when its parameters are constrained in a certain way (the hopping amplitude $t$ is equal to the $p$-wave pairing gap $\Delta$).

The paper is divided in the following way. In Section \ref{sec_MZM_def} we present a more formal definition of MZMs, and in Section \ref{sec_FF_Hams} we deduce the frustration-free Hamiltonians hosting MZMs, given in Eq.\ \eqref{eq_H'} in the spin picture, and Eqs. \eqref{eq_MZM_Ham_ferm} and \eqref{eq_MZM_Ham_ferm_2} for fermions. During our analysis we encounter an FF Hamiltonian, Eq.\ \eqref{eq_H_case_iii}, that supports MZMs, however, they are not spatially separated. Hence, even though these fermionic modes satisfy most of the conditions for them to be MZMs, they are not topologically protected. We comment on how this difference manifests itself in the spin picture in Section \ref{sec_comp_spin_ferm}. Section \ref{sec_MPS} analyses the MZM Hamiltonian Eq.\  \eqref{eq_H'} and its ground space Eq.\ \eqref{eq_case_i_GL} using the language of matrix product states. In Section \ref{sec_broader} we briefly discuss FF Hamiltonians that arise from dropping certain symmetries, and we conclude our work in Section \ref{sec_conc}. Details of calculations are deferred to the Appendices.

\section{\label{sec_MZM_def}Majorana zero modes} 

Consider a Hamiltonian $H$ for a {chain} of (complex) fermions where the operators $c_j^\dag$ and $c_j$ create and annihilate, respectively, a fermion at site $j$. They obey the standard fermionic anti-commutation relations $\{c_j,c_k\}=0$, $\{c^\dag_j,c^\dag_k\}=0$, and $\{c_j,c^\dag_k\}=\delta_{jk}$. The term $n_j = c_j^\dag c_j$ is the number operator. Majorana operators can be defined as
\begin{align}
a_j = c_j + c_j^\dag, \quad b_j = -i(c_j - c_j^\dag),
\end{align}
so that they are Hermitian, and it is possible to rewrite the $H$ in terms of the $a_j, b_j$ with $j=1,\dots,L$. 

For some gapped, fermionic Hamiltonians (such as for superconducting nanowire systems), gapless edge modes called Majorana zero-energy edge modes (MZM), which are a special type of Majorana fermion, may be present. Their existence implies a degenerate ground space, which is protected by particle-hole symmetry. In the thermodynamic limit, the the gap between the ground states vanishes, whereas in finite systems the gap is (at most) exponentially small in the size of the chain. The MZMs cannot be gapped out by any local symmetry-preserving perturbations; they are a manifestation of topological order in the system.

An MZM is a fermionic operator $\gamma$ that satisfies the following conditions \cite{Fendley}:
\begin{enumerate}[(a)]
\item $\gamma^\dag = \gamma$,
\item $\{(-1)^F, \gamma\} = 0$, where $(-1)^F$ is the fermion parity operator and $F = \sum_{j=1}^L n_j$,
\item $[\gamma, H] = 0$,
\item ``normalisable'', that is, $\gamma^2=\I$ as $L \rightarrow \infty$,
\item $\gamma$ is localised near an edge.
\end{enumerate}
Condition (a) ensures that the operator corresponds to a Majorana (i.e. real) fermion, then (b) says that $\gamma$ maps between even parity and odd parity sectors, while (c) implies that the spectrum in these sectors are identical. Condition (d) ensures that the mode is normalisable in the thermodynamic limit, and finally (e) is necessary for MZMs to be topologically protected i.e. robust when subjected to local noise. However, in section \ref{sec_case_iii} we find a setting where conditions (a) through (e) are met, but the MZMs can be gapped by a local perturbation because they are not spatially separated (they are localised on the same end of the chain), so this condition needs to be strengthened.

MZMs always occur in pairs, since each is, in a sense, ``half a fermion'', and $2n$ MZMs implies the existence of a $2^n$-dimensional degenerate ground space. Note that condition (c) may be relaxed so that $\gamma$ does not exactly commute with the $H$ for all $L$, but only needs to as $L \rightarrow \infty$ \cite{alicea2012, SFN_2015}. For a single Kitaev chain, $n=1$ \cite{Kitaev}.

A symmetry-protected topologically ordered (SPTO) phase is a phase that occurs because of some special symmetry in the system. The Kitaev chain has three symmetries: (i) fermionic parity; (ii) particle-hole; (iii) time reversal. These symmetries are rather natural because: (i) parity is preserved in all fermionic systems; (ii) particle-hole symmetry arises naturally in superconducting systems (though it is crucial for protecting the ground space degeneracy); (iii) it turns out that time-reversal symmetry is not necessary for gapless boundary modes to exist \cite{Bernevig_lecs}.   Noting this, in the following we will not use the term SPTO, and instead use the phrases ``topological order'' or ``topological phase'' in reference to the non-trivial phase of the interacting Kitaev chain which hosts MZMs.

\section{\label{sec_FF_Hams}Frustration-free Hamiltonians}

We begin our analysis in the spin-1/2 (qubit) picture in order to connect with the work of Bravyi and Gosset in  \cite{Bravyi_Gosset}. We are concerned with a one-dimensional chain of $L \geq 2$ qubits, and its Hamiltonian is translationally-invariant with nearest-neighbour dimer interactions and open boundaries. Furthermore this Hamiltonian is frustration-free (FF). In other words, the Hamiltonian we focus on is given by
\begin{align}
H=\sum_{i=1}^{L-1} h_{i,i+1}.
\end{align}
Without loss of generality, we can set the ground state energy of each dimer term $h_{i,i+1}$ equal to zero. Then $H$ is FF if a ground state $\ket\Omega$ of $H$ is simultaneously a ground state of every dimer term, i.e., $h_{i,i+1}\ket\Omega = 0$ for all $i = 1,\dots L$. This also means that the ground state energy of $H$ is zero. 
Open boundary conditions are required so that the system has the potential to realise MZMs at the ends of the chain.

As $h:=h_{i,i+1}$ is an operator on $\mathbb{C}^2 \otimes \mathbb{C}^2$ with minimum eigenvalue zero, then $h\geq0$ and $\mathrm{rank}(h) = r$ with $r \in \{1,2,3\}$, hence its ground space is $(4-r)$-fold degenerate. Let its spectral decomposition be
\begin{align}
h= \sum_{j=1}^r \lambda_j \ket{e_j}\bra{e_j},
\end{align}
with eigenvalues $\lambda_1 \geq \lambda_2 \geq \lambda_3 \geq \lambda_4= 0$ (and $\lambda_1 > 0$), and corresponding eigenvectors $\ket{e_j}$. The number of non-zero $\lambda_j$ is equal to $r$. In a frustration-free system, the actual values of the strictly-positive $\lambda_j$ have no effect on the ground space of $H$ or whether it is gapped/gapless (see Appendix \ref{app_mu_gap}). Therefore, we could in principle set all $\lambda_j > 0$ to unity to simplify the task, as is done in \cite{Bravyi_Gosset} and, for instance, the quantum 2-SAT problem \cite{BOE_2010,Bravyi_2SAT}. However, for our purposes we do not impose this restriction, and, in fact, if we did, we would miss out all the interacting Hamiltonians!

We now proceed with our investigation into whether the chain Hamiltonian $H$ possesses Majorana edge modes. The requirement that $H$ be FF depends on the rank $r$ of the dimer term $h$, and, in most cases, the ground space is related to the (qubit permutation) symmetric subspace. We begin with the rank $r=1$ and $r=3$ cases because they are the simplest to handle. The rank 2 case is more involved and so we deal with it last.

\subsection{\label{sec_Rank_1}The dimer terms $h_{i,i+1}$ are rank-one operators}

If the rank of $h$ is 1 and $h = \ket{e_1}\bra{e_1}$ with $\ket{e_1}\in \mathbb{C}^2 \otimes \mathbb{C}^2$ entangled, that is $\ket{e_1}$ cannot be separated into a tensor product of single qubit states, then $H$ is FF and the ground space $G_L$ of $H$ is isomorphic to the symmetric subspace on $L$ qubits, which is $(L+1)$-dimensional \cite{Bravyi_Gosset,Arams_church}. The rule governing whether the Hamiltonian $H=\sum_{i=1}^{L-1} \ket{e_1}\bra{e_1}_{i,i+1}$ is gapped is the main result of \cite{Bravyi_Gosset}, and it applies to a special set of entangled states $\ket{e_1}$. Within this set, the family of fermionic parity conserving Hamiltonians has (up to constants and scaling) $h_{i,i+1} = \cos\theta(Z_i - Z_{i+1}) + \sin\theta(X_iX_{i+1} + Y_{i}Y_{i+1})-Z_i Z_{i+1}$, where $\theta \in (0,\pi/2)$ (see Appendix \ref{App_rank1} for details). Therefore any Hamiltonian in this family is a candidate for supporting MZMs. The problem however is in the degeneracy of the ground space $G_L$. An argument for why there are no MZMs goes as follows. A ground space with dimension $2^n$ is a necessary condition for the existence of $2n$ MZMs. Since $\mathrm{dim}(G_L) = L+1$, we see that the system may have $2\log_2(L+1)$ MZMs if $\log_2(L+1)$ is an integer. {However, the number of MZMs should not depend on the size of the chain, since this is not a topological property.} Furthermore, in the thermodynamic limit $L\rightarrow\infty$, if there were MZMs, there would be a macroscopic number of them. This is an unstable setting and one would not expect these MZMs to be protected against local perturbations since a large number of these zero-energy modes would be overlapping. We thus lose the exponential separation of modes and they would be easily gapped out by local operations.

\subsection{\label{sec_Rank_3}The dimer terms $h_{i,i+1}$ are rank-three operators}

If the rank of $h$ is 3, then $H$ is FF if and only if the ground state of $h$ is a product state of the form $\ket\psi^{\otimes 2}$. It follows that ground space of $H$ is the one-dimensional span of $\ket\psi^{\otimes L}$ \cite{Bravyi_Gosset}. 
A necessary condition for the presence of MZMs is a degenerate ground space, therefore there can be no MZMs in this case. 

\subsection{\label{sec_rank2}The dimer terms $h_{i,i+1}$ are rank-two operators}

In this case, the Hamiltonian $H$ we are considering has nearest-neighbour dimer terms
\begin{align}
h=\lambda_1\ket{e_1}\bra{e_1} + \lambda_2\ket{e_2}\bra{e_2}.
\end{align}
We also assume $h$ is not a product operator i.e. $h\neq  h_A \otimes h_B$. Theorem 3 of \cite{Bravyi_Gosset} states that imposing frustration-freeness on $H$ results in five independent cases. Two of these cases are irrelevant for our purposes (they are numbered 1 and 5 in Theorem 3 of \cite{Bravyi_Gosset}) as they correspond to the ground space being 1- and 0-dimensional, respectively. Such a system cannot host MZMs. We therefore turn our attention to the remaining three cases, which we summarise here in a form more suited to our needs.

Let $G_L$ be the ground space of FF Hamiltonian $H$ with $L$ sites, then $\dim(G_L) = 2$ and exactly one of the following holds:
\begin{enumerate}[(i)]
\item $G_L = \mathrm{span}\{\ket{\alpha}^{\otimes L},\ket{\beta}^{\otimes L}\}$  for some linearly independent normalised single-qubit states $\ket\alpha, \ket\beta$.
\item $G_L = \mathrm{span}\{\ket{\alpha \beta\alpha\beta...},\ket{\beta \alpha\beta\alpha...}\}$  for some linearly independent normalised single-qubit states $\ket\alpha, \ket\beta$.
\item $G_L = \mathrm{span}\{\ket{\alpha}^{\otimes L},\,\,\ket{\bar{\alpha} \alpha\alpha...\alpha}+f\ket{ \alpha\bar{\alpha}\alpha...\alpha}+f^2\ket{ \alpha\alpha\bar{\alpha}...\alpha}+...+f^{L-1}\ket{ \alpha\alpha\alpha...\bar{\alpha}}\}$  for some orthonormal single-qubit states $\ket\alpha, \ket{\bar{\alpha} }$, a non-zero $f\in \mathbb{C}$.
\end{enumerate}
$H$ is gapped in cases (i) and (ii), and it is gapped in case (iii) if $|f| \neq 1$. If $|f| = 1$ in case (iii) then the Hamiltonian is gapless, with spectral gap upper bounded by $(1- \cos(\pi/L))$.

As the ground space in cases (i) - (iii) above is two-fold degenerate for any length $L$, these systems have the potential for MZMs. We now explore each of these three cases in more detail.

\subsubsection{\label{sec_case_i}Case (i): $G_L = \mathrm{span}\{\ket{\alpha}^{\otimes L},\ket{\beta}^{\otimes L}\}$}

This ground space $G_L = \mathrm{span}\{\ket{\alpha}^{\otimes L},\ket{\beta}^{\otimes L}\}$ is a two-dimensional subspace of the $L$-qubit symmetric subspace. In Appendix \ref{app_rank2}, we work in the spin picture to deduce the translationally invariant FF Hamiltonians for which $G_L$ is the ground space. We then impose parity conservation in order to identify this with a fermionic Hamiltonian. In the spin picture, parity conservation is the requirement that $[H,Z^{\otimes L}] = 0$, where $Z= \ket 0 \bra 0 - \ket 1 \bra 1$ is the Pauli-$Z$ matrix. This splits our the Hamiltonians into two distinct types. Writing the dimer terms as
\begin{align}h_{i,i+1} = [A\ket\Psi\bra\Psi + B \ket\Phi\bra\Phi]_{i,i+1}, \quad A,B > 0,\end{align}
where $\ket\Psi$ and $\ket\Phi$ are orthogonal, the two types of Hamiltonian, up to symmetry-preserving unitaries of the form $U_1\otimes U_2 \otimes \dots \otimes U_L$, with $U_i \in$ SU(2), are:
\begin{itemize}
\item \emph{Type 1}
\begin{align}
\ket\Psi &= \frac{1}{\sqrt 2}(\ket{01}-\ket{10}),\\
\ket\Phi & = \cos \frac \omega 2\ket{00} +  \sin \frac \omega 2\ket{11},
\end{align}
where $\omega \in (0,\pi)$. The ground space states are
\begin{align}\label{alpha}
\ket{\alpha} &= \cos\frac \theta 2\ket 0 + i\sin\frac \theta 2\ket 1,\\
\label{beta}
\ket {\beta} &= \cos\frac \theta 2\ket 0 - i\sin\frac \theta 2\ket 1,
\end{align}
with $\theta \in (0,\pi)$ and the angles $\theta, \omega$ are related through
\begin{align}\cos \frac \omega 2 = \frac{\sin^2 \frac \theta 2}{\sqrt{\sin^4 \frac \theta 2+ \cos^4 \frac \theta 2}}
.\end{align}
\item \emph{Type 2}
\begin{align}\label{eq_Type2_Psi}
\ket\Psi &= \cos\frac\theta 2 \ket{01}+\sin\frac\gamma 2 \ket{10},\\
\label{eq_Type2_Phi}
\ket\Phi &=  \sin\frac\theta 2 \ket{01}- \cos\frac\gamma 2 \ket{10},
\end{align}
with $\theta \in (0,\pi)$.
The ground space states are $\ket\alpha = \ket 0$ and $\ket\beta = \ket 1$.
\end{itemize}
\bigskip

Consider first a Type 1 Hamiltonian. The ground space is
\begin{align}\label{eq_case_i_GL}
G_L = \left\{ \left(\cos\frac \theta 2\ket 0 + i\sin\frac \theta 2\ket 1\right)^{\otimes L}, \,\,  \left(\cos\frac \theta 2\ket 0 - i\sin\frac \theta 2\ket 1\right)^{\otimes L} \right\}.
\end{align}
The total Hamiltonian may be written as $H =\frac14[(L-1)(A+B)\I+ H'],$ where
\begin{align}\label{eq_H'}\nonumber
H' = \sum_{i=1}^{L-1}& B\cos\omega(Z_i + Z_{i+1})-(A-B\sin\omega)X_iX_{i+1}\\
&-(A+B\sin\omega)Y_iY_{i+1}-(A-B)Z_iZ_{i+1},
\end{align}
and $X = \ket 0 \bra 1 + \ket 1\bra 0$, $Y=-i\ket 0 \bra 1 + i \ket 1 \bra 0$ are the Pauli-$X$ and Pauli-$Y$ matrices respectively.
This is a Heisenberg $XYZ$ spin-chain with a local magnetic field.
As $H$ and $H'$ are related simply by a rescaling and constant shift, they have the same relevant physical properties. From now on we only consider $H'$.
Performing a Jordan-Wigner transformation (see Appendix \ref{app_JW}) on $H'$ gives the Hamiltonian for a superconducting Kitaev chain of interacting spinless (complex) fermions
\begin{align}\nonumber
H'= &\sum_{j=1}^{L-1}[-t(c^\dag_j c_{j+1}+c^\dag_{j+1}c_j)+\Delta(c_j c_{j+1}+c^\dag_{j+1}c^\dag_j)]\\ 
&-\frac12\sum_{j=1}^L\mu_j(2n_j-\I)+U\sum_{j=1}^{L-1}(2n_j-\I)(2n_{j+1}-\I),\label{eq_MZM_Ham_ferm}
\end{align}
which is time-reversal and particle-hole symmetric (although it does not conserve total fermion number).
The operators $c_j^\dag$ and $c_j$ create and annihilate, respectively, a fermion at site $j$, and they obey the standard fermionic anti-commutation relations. The term $n_j = c_j^\dag c_j$ is the number operator, $t$ is the hopping amplitude, $\Delta$ is the $p$-wave pairing gap, $\mu_j$ is the on-site chemical potential, and $U$ is strength of the nearest-neighbour interaction. These terms are related to the parameters in Eq.\ \eqref{eq_H'} in the following way:
\begin{align}\label{eq_MZM_Ham_ferm_2}
t=2A, \quad \Delta=-2B\sin\omega, \quad U=B-A,
\end{align}
and the boundary $\mu_j=2B\cos\omega$ for $j=1$ and $j=L$, which is half the bulk value $\mu_j=4B\cos\omega$ for $j=2,\dots,L-1$. Alternatively one can eliminate the spin Hamiltonian variables $A,B,\omega$ to obtain
\begin{align}\label{eq_MZM_Ham_ferm_3}
\mu = 4\sqrt{U^2+tU+\frac 14(t^2-\Delta^2)},
\end{align}
where $\mu$ is the bulk value. This expression is in agreement with Eq. (11) of Ref.\ \cite{Katsura_etal}. These relations ensure that $H'$ is frustration free. Given $A,B > 0$ and $\omega \in (0,\pi)$, this restricts $t>0$ and $\Delta < 0$. The sign of $\Delta$ is simply due to the local basis fixing in Eqs. \eqref{alpha}  and \eqref{beta}. Transforming every qubit by the unitary matrix $V = \ket{0}\bra{0} + i\ket{1}\bra{1}$ has the effect of sending $\omega \rightarrow -\omega$, and therefore changes the sign of $\Delta$ (only). In the fermion picture, the same transformation is achieved by $c_j \rightarrow ic_j$. As $U = B-A$, the fermion interactions can be either repulsive $U > 0$, non-interacting $U=0$, or attractive $U < 0$. It only depends on relative sizes of the eigenvalues $A,B$ of the dimer Hamiltonian $h_{i,i+1}$. 

The ground states of $H'$ in the fermion picture and their indistinguishability with respect to local measurements are discussed in Ref.\ \cite{Katsura_etal}. In the spin picture, the form of the ground space $G_L$ in Eq.\ \eqref{eq_case_i_GL} is not in an amenable form since the products states are not orthonormal and they are not of definite parity. Yet the ground space should split into even and odd parity sectors. In fact, this is easily achieved by taking appropriate linear combinations:
\begin{align}
\ket\alpha^{\otimes L} + \ket\beta^{\otimes L}, \quad \mathrm{and} \quad \ket\alpha^{\otimes L} - \ket\beta^{\otimes L}.
\end{align}
These vectors are now orthogonal and are parity eigenstates with eigenvalues $+1$ and $-1$ since $Z\ket\beta = \ket\alpha$.

So far our FF Hamiltonian $H'$  is gapped, conserves fermionic parity, and has a ground space degeneracy. However for Majorana zero modes to exist, it is also compulsory that $H'$ is in a topologically non-trivial phase. In \cite{Kitaev}, this phase is defined for a non-interacting Hamiltonian. We now demonstrate that our $H'$ is adiabatically connected to the non-interacting case by showing that the system remains gapped and the ground states do not change, hence they are in the same phase. The argument below follows the one presented in \cite{Katsura_etal}. We reproduce it here for completeness.

Define $s = (B-A)/2A > -\frac 12$, then the family of Hamiltonians $H'(s) = \sum_{i=1}^{L-1}h'(s)_{i,i+1}$, with 
\begin{align}
h'(s)_{i,i+1} &= -2A[(c^\dag_i c_{i+1}+c^\dag_{i+1}c_i)+(1+2s)\sin\omega(c_i c_{i+1}+c^\dag_{i+1}c^\dag_i)\\
&+(1+2s)\cos\omega(n_i+ n_{i+1}-\I) -s(2n_i-\I)(2n_{i+1}-\I)], \notag
\end{align} 
all have the same ground space $G_L$ given in Eq.\ \eqref{eq_case_i_GL}.
The parameter $s$ interpolates between interacting systems with $s \neq 0$ and the non-interacting case $s=0$, i.e. $A=B$. The special case $H'(s=0)$ corresponds to the Kitaev chain \cite{Kitaev} in the topologically-ordered regime: topological order is present in the system only when $2|t| > |\mu|$ (where $\mu$ corresponds to the bulk on-site chemical potential value $\mu=4B\cos\omega$) and $\Delta \neq 0$. In terms of our parameters, these two conditions are $1 > \cos\omega$ and $\sin\omega \neq 0$, which are clearly true for the full range of $\omega \in (0,\pi)$. Hence $H'(s=0)$ is topologically-ordered quite generally, that is, for all allowed $\omega$, and it is known that $H'(s=0)$ supports MZMs localised at the ends of the chain \cite{Kitaev,Katsura_etal}. Moreover, from \cite{Bravyi_Gosset} and Appendix \ref{app_mu_gap}, we know that the Hamiltonian $H'(s)$ remains gapped for all $s > -\frac 12$ (i.e. for all $A,B > 0$). Therefore, since the interacting system is adiabatically connected to the non-interacting one, they are in the same topologically non-trivial phase.

The condition that the gap does not close along the whole path $s$ may not be sufficient to ensure $H'(s)$ and $H'(0)$ are in the same phase. This is because interactions modify the topological classifications of fermionic Hamiltonians, and a Hamiltonian in the topological phase may be connected to a trivial one without the gap closing \cite{FidKit_1}. Without interactions, these Hamiltonians are indexed by an integer $k \in \mathbb{Z}$ (the ``topological invariant''); adding interactions modifies this to $k \in \mathbb{Z}$ modulo 8 \cite{FidKit_2,Ari_class}. 
The Kitaev chain $H'(0)$ is in the Altland-Zirnbauer symmetry class BDI \cite{AZ_class} because the Hamiltonian commutes with the time-reversal operator $\mathcal T$, with $\mathcal T^2 = 1$, and with the charge conjugation operator (due to particle-hole symmetry) $\mathcal P$, with $\mathcal P^2 = 1$. In the Kitaev chain, the topological invariant $k=1$ \cite{Bernevig_lecs}. To study interacting models, one considers $n$ parallel Kitaev chains, i.e. $2n$ MZMs $\gamma_i$, and then asks what kinds of symmetry-preserving interaction terms can gap out the MZMs and lead to a non-degenerate ground space. It turns out that 8 chains are needed, i.e. an interaction involving 8 MZMs, before such an interaction can arise (hence the modification of $\mathbb{Z}$ to $\mathbb{Z}$ modulo 8 mentioned above). For all $s$, our interacting Hamiltonian $H'(s)$ is still a single chain of fermions, which retains a gap and the same symmetries. Therefore there is no way of gapping out the MZMs on this chain, and so the ground space remains degenerate (in fact the ground space is independent of $s$). Hence, $H'(s)$ and $H'(0)$ are the same non-trivial phase.

Further evidence for topological order in $H'(s)$ is supplied in \cite{Katsura_etal} by identifying two fermionic operators $\gamma_L$ and $\gamma_R$ that satisfy conditions (a) to (e) in section \ref{sec_MZM_def} for $H'(0)$ and so are candidate MZMs for the non-interacting system. In addition they correspond to modes localised on opposite ends of the chain. However, $[\gamma_{L,R}, H'(s)] \neq 0$ for $s \neq 0$. Therefore, an analytical form for MZMs in the interacting case is yet to be found \cite{Katsura_etal}.
\bigskip\bigskip

Consider now a Type 2 Hamiltonian characterised by Eqs.~\eqref{eq_Type2_Psi} and \eqref{eq_Type2_Phi}. It is equal to $H = \frac 14[(L-1)(A+B)\I + H' ]$ where
\begin{align}\label{eq_case_i_Type2}
H'=\sum_{i=1}^{L-1} (A-B)[\cos \theta (Z_i-Z_{i+1}) +\sin\theta (X_iX_{i+1} + Y_iY_{i+1})]-(A+B)Z_iZ_{i+1}.
\end{align}
This is adiabatically connected to an Ising Hamiltonian (where $A=B$) with dimer terms $Z_i Z_{i+1}$. In the fermionic picture this has $t= \Delta = \mu = 0$, therefore this Hamiltonian is in the topologically trivial phase and there are no MZMs.

\subsubsection{\label{sec_case_ii}Case (ii): $G_L = \mathrm{span}\{\ket{\alpha \beta\alpha\beta...},\ket{\beta \alpha\beta\alpha...}\}$}

Given any two states $\ket\alpha$ and $\ket\beta$, it is always possible to find a unitary $U$ such that $U\ket\alpha = \ket\beta$ and $U^2 = \I$. This means that case (ii) is locally unitarily related to case (i), and the derivation of the case (i) Hamiltonian in Appendix \ref{app_rank2} can be easily modified to account for case (ii).

The topologically non-trivial case (i) Hamiltonian is of Type 1. Since $\ket\beta = Z\ket\alpha$ and $Z^2 = \I$, then case (ii) is related to case (i) by a local unitary transformation that is a Pauli-$Z$ on every even qubit $Z_{even}=\I \otimes Z \otimes \I \otimes Z \dots$ or odd qubit $Z_{odd}=Z \otimes \I \otimes Z \otimes \I \dots$. Let $\bar Z \in \{Z_{even},Z_{odd}\}$. Then the Hamiltonian of interest in case (ii) is $\bar Z  H' \bar Z $, where $H'$ is given by Eq.\ \eqref{eq_H'}. The Hamiltonian remains translationally invariant and preserves fermionic parity. In the fermion picture, $\bar Z$ has the effect of sending $t \rightarrow -t$ and $\Delta \rightarrow -\Delta$, and corresponds to the local unitary transformation $c_j \rightarrow (-1)^jc_j$.
(Note that if we only wanted to transform $t\rightarrow -t$ and change nothing else, this would correspond to a combination of the above unitaries: in the spin picture it is $V^{\otimes L}\bar Z $, where $V = \ket 0 \bra 0+ i \ket 1 \bra 1$, and in the fermion picture $c_j \rightarrow i(-1)^jc_j$.)
As case (i) and case (ii) are related by a local unitary that commutes with the parity operator, $\bar Z  H' \bar Z $ has all the same physical properties as $H'$. In other words, it has the required symmetries and is also in the topologically non-trivial phase and can host MZMs localised at opposite ends of the chain.

\subsubsection{\label{sec_case_iii}Case (iii): 
$G_L = \mathrm{span}\{\ket{\alpha}^{\otimes L},\,\,\,\ket{\bar{\alpha} \alpha\alpha...\alpha}+f\ket{ \alpha\bar{\alpha}\alpha...\alpha}+f^2\ket{ \alpha\alpha\bar{\alpha}...\alpha}+...+f^{L-1}\ket{ \alpha\alpha\alpha...\bar{\alpha}}\}$
}

Let $\ket\alpha,\ket{\bar \alpha}\in \mathbb{C}^2$ be a pair of orthonormal qubit states. In \cite{Bravyi_Gosset}, it is shown that the Hamiltonian
\begin{align}
H = \sum_{i=1}^{L-1} A\ket{\bar \alpha\bar \alpha}\bra{\bar \alpha\bar \alpha}_{i,i+1} + B\ket{\nu}\bra{\nu}_{i,i+1},\quad A,B > 0
\end{align}
with $\ket{\nu} =  (\ket{\alpha\bar \alpha} - f\ket{\bar \alpha\alpha})/\sqrt{1+|f|^2}$ and non-zero $f\in \mathbb{C}$ has the two-fold degenerate ground space $G_L= \mathrm{span}\{\ket{\alpha}^{\otimes L},\,\,\,\ket{\bar{\alpha} \alpha\alpha...\alpha}+f\ket{ \alpha\bar{\alpha}\alpha...\alpha}+f^2\ket{ \alpha\alpha\bar{\alpha}...\alpha}+...+f^{L-1}\ket{ \alpha\alpha\alpha...\bar{\alpha}}\}$. In fact, a more general FF Hamiltonian with this ground space is permissible, and that is one that is a sum of dimer terms $[U(A\ket{\bar \alpha\bar \alpha}\bra{\bar \alpha\bar \alpha} + B\ket{\nu}\bra{\nu})U^\dag]_{i,i+1}$ where $U$ rotates only in the subspace $\mathrm{span}\{\ket{\bar \alpha\bar \alpha},\ket\nu\}$.  The Hamiltonian is gapped when $|f|\neq 1$. 
Imposing parity conservation on $H$ fixes $\ket \alpha = \ket 0$, $\ket{\bar \alpha}=\ket 1$, $U=\I$, and it is possible to choose a local basis such that $f \in \mathbb{R}$.
Then the ground states of $H$ are $\ket{0}^{\otimes L}$ and
{
\begin{align}
\ket{1000 \dots 00} + f \ket{0100 \dots 00} +\dots 
+ f^{L-2} \ket{0000 \dots 10} + f^{L-1} \ket{0000 \dots 01},
\end{align}
}which is left unnormalised. Since $|f|\neq 1$, this ground state is not permutation symmetric.
The Hamiltonian is given by $H = \frac {1}4\left[(L-1)\left([1+f^2]^{-1}B+A \right)\I+H'\right],$ with $H' = \sum_{i=1}^{L-1}h'_{i,i+1}$ and
{
\begin{align}
h'_{i,i+1}=&-\left(A-\frac{B(1-f^2)}{1+f^2}\right)Z_i - \left(A+\frac{B(1-f^2)}{1+f^2}\right)Z_{i+1} \notag \\
&-\frac{2Bf}{1+f^2}(X_iX_{i+1}+Y_iY_{i+1}) 
+(A-B)Z_iZ_{i+1}.
\end{align}
}
The fermionic dimer Hamiltonian (see Appendix \ref{app_JW})  is
{
\begin{align} \nonumber
h'_{i,i+1}=&\left(A-\frac{B(1-f^2)}{1+f^2}\right)(2n_i-\I) 
+\left(A+\frac{B(1-f^2)}{1+f^2}\right)(2n_{i+1}-\I)\\ 
&-\frac{4Bf}{1+f^2}(c^\dag_ic_{i+1}+c^{\dag}_{i+1}c_i)
+(A-B)(2n_{i}-\I)(2n_{i+1}-\I).\label{eq_H_case_iii}
\end{align}
}  
Following the same procedure as in case (i), we introduce a variable $s \propto A-B$ such that $H'(s)$ is a one parameter family of Hamiltonians which is adiabatically connected to the non-interacting system $H'(s=0)$. As long as we fix $|f| \neq 1$ the system remains gapped and the ground space is constant as we vary $s$, so $H'(s)$ and  $H'(s=0)$ are in the same phase. Is the Hamiltonian $H'(s=0)$ topologically non-trivial? The two necessary conditions from \cite{Kitaev} for this are $\Delta\neq 0$ and $2t > |\mu|$. Immediately we see that $\Delta= 0$, which is already bad news. Furthermore $2t > |\mu|$ from \cite{Kitaev} turns out to be false. To see this, note that the bulk $\mu = 4A$, and with $t = \frac{4Bf}{1+f^2}$, the inequality leads to $(1-f)^2 < 0$ which is false for all real $f$. Therefore, $H'(s)$ is adiabatically connected to a non-interacting Hamiltonian in the trivial phase. Now, while the non-interacting Hamiltonian $H'(s=0)$ may not host MZMs, it does not immediately preclude the possibility that the interacting one may. If $H'(s=0)$ is in the trivial phase, then its topological invariant $k=0$ \cite{Bernevig_lecs}. When we switch on interactions, the topological invariant of the interacting system goes to $k = 0 \mod 8 = 0$ \cite{FidKit_1,Ari_class}, hence the interacting system is also trivial and so it does not support MZMs.

This model does, however,  possess zero modes due to the way it was constructed.  These are easiest to analyze in the non-interacting limit $A=B$. The Hamiltonian becomes quadratic in the fermionic operators $c_j, c^\dag_j$, and so we can express it as $H'(0)= \frac 12 Q^\dag W Q$ where $Q = (c_1,\dots,c_L)^T$. In this case, $W$ is a Hermitian tri-diagonal $L \times L$ matrix, and its null vector has the form $\ve u = (f, f^2, f^3, \dots,f^L)^T$. 
Hence one finds that the complex fermionic zero mode is  $\wt{c} = {\cal N} \sum_{j=1}^L f^j c_j$, where ${\cal N}$ is a normalization constant. This mode commutes with the Hamiltonian: $[\wt{c}, H]=0$.  When $|f|<1$ this corresponds to an edge mode localised on the left side of the chain while for $|f| >1$ the mode is localized on the right.  The two Majorana fermions, $\gamma_1 = \wt c + \wt c^\dag$ and  $\gamma_2 = -i(\wt c - \wt c^\dag)$, composing $\tilde{c}$ each satisfy the conditions set out in section \ref{sec_MZM_def} and therefore correspond to zero modes, however they are localised at the same end. Since these modes are not spatially separated, the Majoranas can be easily gapped by a local perturbation.

There is another way to interpret the topological triviality of the Hamiltonian in Eq. \eqref{eq_H_case_iii} by noticing that it preserves fermion number (not just fermion parity). Then the degenerate ground states are states of definite fermion number, 0 or 1 (and opposite parity). These are distinguishable by the local operator $\hat N  =\sum_{i=1}^{L} n_i$, and hence they do not enjoy the topological protection that Hamiltonians of case (i) do. The ground states of the case (i) Hamiltonians are superpositions of all number states of definite parity, and therefore suffer huge fluctuations in measurements of $\hat N$, and so effectively, when $L$ is large, one cannot distinguish the ground states since $\langle \hat N \rangle \approx \frac L2$ for both states. Notice that the rank 1 Hamiltonian of section \ref{sec_Rank_1}, and the case (i) type 2 Hamiltonian of equation Eq.~\eqref{eq_case_i_Type2}, also conserve fermion number and are topologically trivial. Nevertheless, because number conservation is a natural symmetry in several systems, proposals for constructing MZM model within number-conserving systems have been presented in Refs.~\cite{Lang_2015,Iemini_2015}. There the authors consider two coupled Kitaev chains and demonstrate topological properties by analytically diagonalising the full Hamiltonians, which are frustration-free.

\subsection{\label{sec_comp_spin_ferm}Comparison of the spin and fermion pictures}

There are a few observations to be made that arise from switching between the spin and the fermion pictures. Consider first the spin picture. A key difference between the cases (i) (or (ii)) and (iii) is that, in the former, cases (i) or (ii), the ground space is invariant if we close the chain. 
That is, making the Hamiltonian in Eq.\ \eqref{eq_H'} periodic by adding a coupling $h_{1,L}$ term between sites $1$ and $L$ does not affect the ground space $G_L$ Eq.\ \eqref{eq_case_i_GL}. In fact, for case (i) $G_L$ is invariant under the addition of a coupling $h_{i,j}$ between any two spins (not just nearest-neighbour). In a sense, case (i) is topologically-trivial in the spin picture. This is in stark contrast to the Majorana picture, where closing the chain results in a loss of MZMs since there is no Majorana operator $\gamma$ that commutes with 
$h_{1,L}$.  Therefore  topologically trivial spin systems may correspond to topologically non-trivial systems of fermions.

On the hand, let us now consider case (iii).   For the spin system, if we close the chain the ground space degeneracy is lost and the only remaining ground state is $\ket 0^{\otimes L}$. The spin ground space now is sensitive to the topology.   The ground state degeneracy of the fermion system is similarly removed by closing the fermion chain.

\section{\label{sec_MPS}Examining the MZM Hamiltonian through the lens of matrix product states}

In section \ref{sec_case_i} we derived the FF Hamiltonian, Eq.\ \eqref{eq_H'}, that supports Majorana zero modes. We now formulate the spin ground states of this Hamiltonian as matrix product states (MPSs). This was explored previously in Ref. \cite{Asoudeh_etal} but there the connection with frustration freeness and other properties was not so explicit. Here we highlight these observations, as well as discussing additional aspects like injectivity, and confirm that the system satisfies an area law.

The matrix product state (MPS) is a particular representation of a quantum state. The MPS for a state of $L$ qudits is
\begin{align}\label{eq_MPS}
\ket\psi = \sum_{i_1,\dots,i_L} \tr(A^{[i_1]}\dots A^{[i_L]})\ket{i_1\dots i_L},
\end{align}
where $i_k \in \{0,\dots,d-1\}$ for all $k = 1\dots L$. The $i_k$ label the physical qudits with dimension $d$ and the $A^{[i_k]}$ are $D_k \times D_{k+1}$ matrices where $D=\max_k D_k$ is the ``bond dimension''. For open boundary conditions $D_1=D_L=1$. Any quantum state can be written in MPS form for large enough $D$, however, the MPS is most useful when $D$ is constant in $L$ since this enables efficient computation of measurable quantities  \cite{MPS_open,Orus}. The MPS formalism is useful for approximating ground states of one dimensional quantum spin models. Both the AKLT \cite{AKLT} and Majumdar-Ghosh \cite{Maj_Ghosh} Hamiltonians have ground states that can be efficiently represented using MPSs. Matrix product states, and their generalisations to tensor networks, are a powerful resource in the study of many body systems. MPSs form the variational domain for the density matrix renormalisation group \cite{MPS_DMRG,MPS_open}, and provide an invaluable tool for analysing area laws \cite{Hastings}. Furthermore, for every MPS, there is a frustration-free ``parent Hamiltonian'' for which the MPS is the ground state \cite{PEPS_inj}. We study this in more detail below, after we recast our MZM Hamiltonian (in the spin picture) ground states as MPSs.

The spin state for which we seek an MPS is any vector in the ground space $G_L$ from Eq.\ \eqref{eq_case_i_GL} i.e.
\begin{align}
\ket\psi \in \mathrm{span}(\ket\alpha^{\otimes L},\ket\beta^{\otimes L}) = G_L,
\end{align}
with the local basis states chosen such that $\ket\alpha=\cos(\theta/2)\ket 0 + i\sin(\theta/2)\ket 1=Z\ket{\beta}$. 
The ground space $G_L$ is a subspace of the symmetric subspace $S_L$ on $L$ qubits, where $\mathrm{dim}(S_L)  =L+1$. All states in $S_L$ can be written as an MPS with diagonal matrices $A^{[i_k]}$ and bond dimension $D=L+1$ \cite{MPS_symm}. The linear growth of $D$ with system size may mean that the area law does not apply when the ground space of a Hamiltonian is $S_L$ (an example of such a FF Hamiltonian is when it is a sum of rank 1 projectors, see Section \ref{sec_Rank_1} above). Nevertheless, it has been shown that this does in fact satisfy an area law \cite{BOE_2010}. In any case, since $\mathrm{dim}(G_L)=2$, the size of the ground space $G_L$ of the MZM Hamiltonian is constant in $L$. Therefore, any state in $G_L \subseteq S_L$ can be represented an MPS with diagonal matrices $A^{[i_k]}$ with $D=2$, and this satisfies an area law \cite{Chubb_Flammia,Huang}. We show in Appendix \ref{app_MPS} that the bond matrices for the state $\ket{\psi} = u\ket\alpha^{\otimes L} + v\ket\beta^{\otimes L} \in G_L$, with the local basis choice $\ket\alpha = \cos(\theta/2)\ket0 + i\sin(\theta/2)\ket 1 = Z\ket\beta$, are
\begin{align}
W^{[0_k]} =  \cos(\theta/2) \I, \quad W^{[1_k]} =  i\sin(\theta/2) Z,
\end{align}
for $k=2,\dots,L-1$, and the boundary matrices are row and column vectors due to open boundary conditions:
$W^{[0_1]} = \cos(\theta/2)(u,v)$, $W^{[1_1]} = i\sin(\theta/2)(u,-v)$, $W^{[0_L]} = \cos(\theta/2)(1,1)^T$, and $W^{[1_L]} = i\sin(\theta/2)(1,-1)^T$.

The parent Hamiltonian $H$ of this MPS is the one we find in Eq.\ \eqref{eq_H'}, which can support MZMs. As its ground state space $G_L$ is degenerate, this MPS is \emph{non-injective}. Such an MPS  corresponds to systems with discrete symmetry breaking \cite{MPS_reps}. An arbitrary state $\ket{\psi} = u\ket\alpha^{\otimes L} + v\ket\beta^{\otimes L} \in G_L$ does not possess the symmetries of Majorana Hamiltonian, which are  fermionic parity conservation, and, additionally, invariance under time reversal. This can be seen from their action in the spin picture, which happens to coincide for the local basis choice: $Z^{\otimes L}(u\ket\alpha^{\otimes L} + v\ket\beta^{\otimes L}) = u\ket{\alpha^*}^{\otimes L} + v\ket{\beta^*}^{\otimes L}=u\ket{\beta}^{\otimes L} + v\ket{\alpha}^{\otimes L} \neq \ket{\psi}$.
If an MPS is, on the other hand, injective, then it is the unique ground state of a parent Hamiltonian, and, in 1D, it is known that this Hamiltonian is gapped \cite{PEPS_inj}. So our MZM Hamiltonian is an example of a non-injective but gapped system.

\section{\label{sec_broader}A broader class of frustration-free Hamiltonians}

In the preceding work, we have only considered FF Hamiltonians with certain symmetries, namely translational invariance and fermionic parity conservation. However, given a ground space $G_L$, there is a whole family of FF Hamiltonians that share this ground space $G_L$. This family is obtained by applying a unitary $U_i$ to each dimer term $h_{i,i+1}$ that only rotates in the range of $h_{i,i+1}$. That is, if $H = \sum_i h_{i,i+1}$ is a Hamiltonian with ground space $G_L$, then so is $H_{U} = \sum_i U_i h_{i,i+1} U_i^\dag$ as long as $\ker(h_{i,i+1}) = \ker(U_i h_{i,i+1} U_i^\dag)$ for all $i$. Notice that the unitary $U_i$ can be site-dependent, i.e. we can drop translation invariance. If we also relax parity conservation then in case (i) we can obtain spin dimer Hamiltonians like Eq.\ \eqref{eq_H'} with an antisymmetric Dzyaloshinskii-Moriya interaction $XZ-ZX$ \cite{Dzya,Moriya}. However, because the system is frustration-free, the coefficients in the dimer Hamiltonian are not independent and so the antisymmetric interaction always appears along with an additional local transverse magnetic field in the $X$ direction.  

\section{\label{sec_conc}Conclusion and future directions}

Employing known results about gapped frustration-free spin system \cite{Bravyi_Gosset}, we use a Jordan-Wigner transformation to deduce the full set of dimer frustration-free fermionic Hamiltonians that can support Majorana zero-energy edge modes. We find that interacting Hamiltonians arise quite generically (the interactions can be either attractive or repulsive), and that they are adiabatically connected to the non-interacting Kitaev chain \cite{Kitaev}, as previously observed in \cite{Katsura_etal} using different methods.
The MZM Hamiltonian in the spin picture corresponds to a Heisenberg $XYZ$ chain with a local magnetic field. We show that the ground states have an efficient MPS representation, and furthermore that the Hamiltonian is non-injective, gapped, and satisfies an area law.

Restricting the Hamiltonian to consist of dimer terms (two-body and nearest neighbour terms $h_{i,i+1}$) in the spin picture is initially an assumption. The resulting Hamiltonians in the fermion picture are of the same form because fermionic parity precludes anything other than fermionic dimer terms. Conversely, a dimer fermionic Hamiltonian only gives rise to dimer spin Hamiltonians. Thus within this setting, our classification of Majorana phases is exhaustive. We focussed on this setting because it allowed us to make conclusive statements, and because it is physically well-motivated as dimer Hamiltonians appear in many experimental settings. One could consider more general terms e.g. (i) three-local terms $h_{ijk}$, or (ii) two-local but not nearest neighbour. Regarding (i), we are not aware of theoretical results classifying the frustration free and gapped regimes for such models, thus this question would need to be addressed first, and it certainly forms an interesting future direction. As for (ii), there may be scope for deriving results for Hamiltonians with terms like $h_{i,i+2}$. The reason is that it turns out that our case (i) spin Hamiltonians (when $t=\Delta$, or equivalently $A = B \sin\omega$) are dual to the ``axial next-nearest neighbour Ising'' (ANNNI) Hamiltonians, and these have terms like $h_{i,i+2}$ \cite{Peschel_Emery}. Because of duality, the phases of our case (i) and the ANNNI match, therefore one could use our methods to analyse Majorana phases in the fermion picture of ANNNI. Also, this set of Hamiltonians happen to lie on the ``one-dimensional line'' of Peschel and Emery \cite{Peschel_Emery}.

The frustration-free requirement could make it difficult to realise our MZM Hamiltonians in an experiment due to the fine-tuning of parameters. Therefore, an analysis of the effects of perturbations needs to be made, and considerations for how the system changes as the Hamiltonian varies away from the FF manifold.

Nevertheless there is still much to be explored within the frustration-free set. Although the ground states of such Hamiltonians may be easy to determine, this does not hold for the excited states, which generally are frustrated. Characterising the whole spectrum would desirable for the purposes of perturbation theory, and also for finding an expression of the MZM for the interacting FF Hamiltonian.

Finally, one can ask about 1D FF Hamiltonians for parafermions, the higher dimensional analogue of MZMs, or for qudits. Already this has been analysed for the non-interacting chain in Ref.~\cite{Fendley}. Perhaps it is tractable also in the interacting case.

\begin{acknowledgments}
The authors would like to thank Ari Turner and Courtney Brell for useful discussions.
SJ is supported by an Imperial College London Junior Research Fellowship. 
RB is supported in part
by the European Union's Seventh Framework Programme for
research, technological development, and demonstration under
Grant No. PCIG-GA-2013-631002.
\end{acknowledgments}

\appendix

\section{\label{app_mu_gap}Gap properties and the spectrum of the dimer terms $h_{i,i+1}$}

\begin{lem}
Consider a system of $L$ spins, where each is associated with an $n$-dimensional Hilbert space $\mathbb C^n$, the total space is $\mathcal H = (\mathbb C^n)^{\otimes L}$. Let $H_L=\sum_{i=1}^{M_L} h^{(i)}$, with $h^{(i)} \neq 0$, be a $k$-local (i.e. each $h^{(i)}$ is bounded and acts non-trivially on at most $k$ spins), frustration-free (FF) Hamiltonian with ground space $G_L$. The upper bound of the sum $M_L$ is some integer that grows with $L$. Without loss of generality we can assume all $h^{(i)}$ to be positive semidefinite and assume lowest eigenvalue(s) of each $h^{(i)}$ to be zero. 
Each term $h^{(i)}$ acts on $l_i \leq k$ spins and has a spectral decomposition $h^{(i)} = \sum_{j=1}^{d_i} \mu^{(i)}_j \ket{e^{(i)}_j}\bra{e_j^{(i)}}$, where $\mu^{(i)}_j \geq 0$, and $d_i = n^{l_i}$. Let $\mathcal J_i \subset \mathcal N_i := \{1,\dots,d_i\}$ denote the set of all indices $j$ for which $\mu^{(i)}_j > 0$, i.e. $h^{(i)} = \sum_{j \in \mathcal J_i} \mu^{(i)}_j \ket{e^{(i)}_j}\bra{e^{(i)}_j}$, and by assumption $\mathcal J_i \neq \emptyset$ and $\mathcal J_i \neq \mathcal N_i$. Define $\wt{H}_L =\sum_{i=1}^{L} \Pi^{(i)}$, where $\Pi^{(i)} = \sum_{j \in \mathcal J_i }  \ket{e^{(i)}_j}\bra{e^{(i)}_j}$ is the projector onto the range of $h^{(i)}$. Let the ground space of $\wt H_L$ be $\wt G_L$. 
Then (a) $\wt H_L$ is FF if and only if $H_L$ is, and $G_L = \wt G_L$; and (b)  $\wt H_L$ is gapped in the thermodynamic limit $L \rightarrow \infty$ if and only if $H_L$ is.
\end{lem}

\begin{proof}
(a) Define
\begin{align}
A^{(i)} = \sum_{j=1}^{d_i} \nu^{(i)}_j \ket{e^{(i)}_j}\bra{e^{(i)}_j},
\end{align}
where $\nu^{(i)}_j = \mu^{(i)}_j$ if $\mu^{(i)}_j > 0$, i.e. $j \in \mathcal J_i$, otherwise $\nu^{(i)}_j = 1$. Then $A^{(i)}$ is invertible and strictly positive, that is, $\bra \omega A^{(i)} \ket \omega > 0$ for all $\ket\omega \neq 0 \in \mathcal H$. Furthermore $h^{(i)} =A^{(i)} \Pi^{(i)}$ for all $i = 1,\dots M_L$. Since $H_L$ is FF, then
\begin{align}
h^{(i)}\ket\phi = A^{(i)}\Pi^{(i)}\ket\phi = 0
\end{align}
for any $\ket\phi \in G_L$ and all $i = 1,\dots M_L$. Now, since $A^{(i)}$ is a (strictly) positive definite operator, then this implies that
the only vector $\ket\omega$ that satisfies $A^{(i)} \ket\omega = 0$ is the null vector, hence we must have $\Pi^{(i)}\ket\phi = 0$ for all $i = 1,\dots L$. But this is precisely the condition for $\wt H_L$ to be FF. Since this holds for any $\ket\phi \in G_L$, we deduce that $G_L \subset \wt G_L$, where $\wt G_L$ is the ground space of $\wt H_L$.

For the converse, we proceed in an analogous manner, and establish that $\wt G_L \subset G_L$. Hence $G_L = \wt G_L$.

(b)  Given $L$, let $s_L$ be the smallest of the non-zero eigenvalue of all the $h^{(i)}$, and $t_L = \mathrm{max}_i \|h^{(i)}\| < \infty$ (where $\|\dots \|$ is the Schatten operator 1-norm) is the largest eigenvalue of all the $h^{(i)}$. Then since $t_L\Pi^{(i)}\geq h^{(i)}\geq s_L \Pi^{(i)}$ for all $i=1,\dots,M_L$, it follows that 
\begin{align}\label{expH_ineqs}
t_L\bra{\psi}\wt H_L\ket{\psi} \geq \bra{\psi}H_L\ket{\psi} \geq s_L\bra{\psi}\wt H_L\ket{\psi},
\end{align}
for all $\ket\psi \in \mathcal H$. Let $E_L$ ($\wt E_L$) be the minimum strictly positive eigenvalue of $H_L$ ($\wt H_L$) with eigenvector $\ket{ E_L}$ ($\ket{\wt E_L}$). Then $\ket {E_L} \in G_L^{\perp}$ and $\ket{\wt E_L} \in \wt G_L^{\perp}$, the orthogonal complements of $G_L$ and $\wt G_L$ respectively. However, note that, since $G_L=\wt G_L$ from (a), then $\wt G_L^{\perp}=G_L^{\perp}$.

First we establish $H_L$ gapped $\Rightarrow$ $\wt H_L$ is gapped. Assume $H_L$ is gapped.  By assumption $E_L$ remains strictly positive as $L \rightarrow \infty$. Then from the first inequality in Eq.\ \eqref{expH_ineqs}, with $\ket\psi = \ket{\wt E}_L$, we have
\begin{align}
t_L\bra{\wt E_L}\wt H_L\ket{\wt E_L} = t_L \wt E_L \geq \bra{\wt E_L}H_L\ket{\wt E_L}
\end{align}
but by definition, the right hand side is lower bounded by $ \bra{E_L}H_L\ket{ E_L} = E_L$.
Hence we have deduced that $t_L \wt E_L \geq E_L$. Since $t_L$ is bounded for all $L$, we must have that $\wt E_L$ remains strictly positive as $L \rightarrow \infty$, hence $\wt H_L$ is gapped if $H_L$ is.

Now we establish $\wt H_L$ gapped $\Rightarrow$ $H_L$ is gapped. Assume $\wt H_L$ is gapped.  By assumption $\wt E_L $ remains strictly positive as $L \rightarrow \infty$. Then from the first inequality in Eq.\ \eqref{expH_ineqs}, with $\ket\psi = \ket{ E_L}$, we have
\begin{align}
\bra{E_L} H\ket{ E_L} = E_L \geq s_L\bra{E_L}\wt H_L\ket{ E_L},
\end{align}
but by definition, the right hand side is lower bounded by $ s_L\bra{\wt E_L}\wt H_L\ket{\wt E_L} = s_L \wt E_L$.
Hence we have deduced that $ E_L \geq s_L \wt E_L$. Since $s_L > 0$ for all $L$, then the right hand side remains strictly positive as $L \rightarrow \infty$, hence $H_L$ is gapped if $\wt H_L$ is.

\end{proof}

\section{\label{App_rank1}Hamiltonian with rank 1 dimer terms}

Consider a general two-qubit state
\begin{align}
\ket\psi = a\ket{00} + b\ket{01} + c\ket{10} + d\ket{11},
\end{align}
and the Hamiltonian $H = \sum_{i=1}^{L-1} h_{i,i+1}$ with $h_{i,i+1} = \ket\psi\bra\psi_{i,i+1}$. To qualify as a valid Majorana zero mode Hamiltonian, it must be preserve fermionic parity. In the spin picture this condition is $[H,Z^{\otimes L}] = 0$, which is equivalent to $[h_{i,i+1},Z_i Z_{i+1}] = 0$ for all $i$, and this can only be satisfied if
\begin{align}
Z \otimes Z \ket\psi = e^{i\phi} \ket\psi, \quad \phi \in \mathbb{R}.
\end{align}
This is an eigenvalue equations for the parity operator $Z \otimes Z$. The eigenvalues of $Z \otimes Z$ are $+ 1$ and $-1$, and the corresponding eigenvectors are even and odd parity respectively. In qubit language, even (odd) parity states are linear combinations of computational basis states with an even (odd) number of 1s. Hence, for two qubits, the even parity sector is spanned by $\{\ket{00},\ket{11}\}$, and odd parity is spanned by $\{\ket{01},\ket{10}\}$.

Hence $\ket\psi$ is either $\ket{\psi_{+}} = a\ket{00} + d\ket{11}$ or $\ket{\psi_{-}} = b\ket{01} + c\ket{10}$. In Ref.\ \cite{Bravyi_Gosset} considers only entangled $\ket\psi$, and it is stated that $H$ is gapped if and only if the matrix
\begin{align}
T_\psi :=
\begin{pmatrix}
\,\,\,\,\braket {\psi}{01} & \,\,\,\,\braket {\psi}{11}\\
-\braket {\psi}{00} & -\braket {\psi}{10}
\end{pmatrix},
\end{align}
has eigenvalues $\lambda_1$ and $\lambda_2$ such that $|\lambda_1| \neq |\lambda_2|$.

The eigenvalues of $T_{\psi_{+}}$ are $\pm a^*d^*$, hence this is gapless.
The eigenvalues of $T_{\psi_{-}}$ are $b^*$ and $-c^*$, hence this is gapped as long as $|b| \neq |c|$. Let $b=\cos\frac \theta 2$ and $c=e^{i\omega}\sin\frac \theta 2$. The product of qubit unitaries $\mathcal S=\otimes_{k=1}^L S_k$, where $S_k = \ket 0\bra 0 + e^{i(k-1)\omega}\ket1\bra1$, commutes with the parity operator $Z^{\otimes L}$, preserves translational invariance and removes the phase $e^{i\omega}$. Hence we can instead consider $c = \sin\frac \theta 2$ without loss of generality. Also required by Ref.\ \cite{Bravyi_Gosset} is $\theta \in (0,\pi/2)$ so that $\ket\psi$ remains entangled. Then
\begin{align}
H &= H(\theta) = \sum_{i=1}^{L-1} \ket\psi\bra\psi_{i,i+1} \notag \\
&= \frac  14[ (L-1) \I + \cos\theta(Z_i - Z_{i+1}) + \sin\theta(X_iX_{i+1} + Y_{i}Y_{i+1})-Z_i Z_{i+1}].
\end{align}

\section{\label{app_rank2}Deriving the FF Hamiltonian for case (i) in section \ref{sec_rank2}}

Recall that $h \geq 0$ is a two-qubit, rank-2 operator which cannot be written as $ h_A \otimes h_B$, for some single qubit operators $h_A, h_B$. We adapt a paragraph from Ref. \cite{Bravyi_Gosset} which shows that the range of $h$ is spanned by two linearly independent states $\ket\psi, \ket\phi$ which are both entangled. Consider the product operator $h_A \otimes h_B$. 
This  is positive and rank-2 if and only if it is of the form $\ket\chi\bra\chi \otimes M$ or $M\otimes \ket\chi\bra\chi$, where $\ket\chi \in \mathbb{C}^2$ and $M$ is a positive definite operator (i.e. rank$(M)=2$). The range of $h_A \otimes h_B$ is then of the form $\mathrm{span}(\ket{\chi}\otimes \ket 0, \ket\chi \otimes \ket 1)$ or $\mathrm{span}(\ket 0 \otimes \ket{\chi}, \ket 1 \otimes \ket{\chi} )$. These are the only two-dimensional subspaces of $\mathbb C^2 \otimes \mathbb C^2$ that contain only product states. Moreover, the rank-2 positive operators with such ranges are always product operators. Since, by assumption, $h \neq h_A \otimes h_B$, then the range of $h$ is not of this form and so it must contain at least one entangled state. Call it $\ket\psi$. Let $\ket\nu$ also be in the range of $h$ but linearly independent from $\ket\psi$. Then $\ket\phi = \ket\psi + z\ket\nu$ with $z \in \mathbb{C}$ is also in the range of $h$, and we can always choose a $z$ such that $\ket\phi$ is entangled and it is linearly independent of $\ket\psi$.

To the state $\ket\psi$ we associate a $2\times 2$ matrix
\begin{align}
T_\psi =
\begin{pmatrix}
\,\,\,\,\braket {\psi}{01} & \,\,\,\,\braket {\psi}{11}\\
-\braket {\psi}{00} & -\braket {\psi}{10}
\end{pmatrix},
\end{align}
and similarly for $\ket\phi$.  Note that
\begin{align}\label{eq_Tpsi_sing}
\ket\psi = [\det(T_\psi)]^* (\I\otimes T^{-\dag}_\psi) \ket\xi,
\end{align}
where $\ket\xi = \ket{01} - \ket{10}$ is the (unnormalised) singlet state, and $T^{-\dag}_\psi := (T^\dag_\psi)^{-1} = (T_\psi^{-1})^\dag$.
Furthermore, matrix $T_\psi$ ($T_\phi$)  is invertible if and only if $\ket\psi$ ($\ket\phi$) is entangled, which it is by assumption.

Recall that the minimum energies of $h_{i,i+1}$ and $H$ are zero, so the ground space is equal to the null space.
In a frustration-free system, an $L$-qubit state is in the null space of $H$ if and only if it is in the null space of $\ket{\psi}\bra\psi_{i,i+1}$ and of $\ket{\phi}\bra\phi_{i,i+1}$ for all $i=1,...,L-1$. The form of the ground states can be presented in terms of the matrices $T_\psi$ and $T_\phi$. In \cite{Bravyi_Gosset}  it is shown that the cases (i)-(iii) in section \ref{sec_rank2} break down into conditions on the eigenvectors of $T^{-1}_\phi T_\psi$ and $T_\psi$.  Below we consider case (i) and work backwards from \cite{Bravyi_Gosset} in order to obtain the FF Hamiltonian whose ground space is $G_L = \mathrm{span}\{\ket{\alpha}^{\otimes L},\ket{\beta}^{\otimes L}\}$. 
Case (ii) is simply a local unitary rotation away from case (i) (see section \ref{sec_case_ii}), and for case (iii) the Hamiltonian is already given in \cite{Bravyi_Gosset}.

In case (i), the matrix $T^{-1}_\phi T_\psi$ has linearly independent eigenvectors $\{ \ket\alpha, \ket\beta\}$ and these are also eigenvectors of $T_\psi$. Hence they are also eigenvectors of $T_\phi$:
\begin{align}\label{T_al_bet}
T^{-1}_\phi T_\psi \ket\alpha \propto \ket\alpha \propto  T^{-1}_\phi  \ket\alpha \quad \Rightarrow T_\phi \ket\alpha \propto \ket\alpha,
\end{align}
and similarly for $\ket\beta$. Note that the assumption $\ket\psi$ and $\ket\phi$ are linearly independent means $T^{-1}_\phi T_\psi$ is not proportional to the identity.
Since $T_\psi$ and $T_\phi$ are $d \times d$  matrices (here $d=2$), and they have a common set of $d$ linearly independent eigenvectors, then these eigenvectors form a (non-orthogonal) basis for $\mathbb{C}^2$. It follows that $T_\psi$ and $T_\phi$ commute, and that they are simultaneously diagonalizable. Let $\{\ket 0, \ket 1\}$ be an orthonormal basis for the qubit space $\mathbb{C}^2$. Define $Q = \ket\alpha\bra0 + \ket\beta\bra 1$ as the matrix whose columns are the eigenvectors $\ket\alpha$ and $\ket\beta$, and $D_\psi$, $D_\phi$ are diagonal matrices of eigenvalues of $T_\psi$ and $T_\phi$ respectively. Then
\begin{align}
T_\psi = Q D_\psi Q^{-1},\\
T_\phi = Q D_\phi Q^{-1}.
\end{align}
From equation \eqref{eq_Tpsi_sing} the entangled vectors in the range of $h$ then are
\begin{align}
\ket\psi &= [\det(D_\psi)]^* \I\otimes (Q^{-\dag})(D^\dag_\psi)^{-1}Q^\dag \ket\xi,\\
\ket\phi &= [\det(D_\phi)]^* \I\otimes (Q^{-\dag})(D^\dag_\phi)^{-1}Q^\dag \ket\xi.
\label{phi_xi}
\end{align}
Using the identity
\begin{align}
Q^{-\dag} = \frac{YQ^*Y}{\det Q^*},
\end{align}
we find $Q^{-\dag} = [\det Q^*]^{-1}(\ket{\bar{\beta}}\bra 0 - \ket{\bar{\alpha}}\bra 1)$, where $\ket{\bar{\alpha}}$ and $\ket{\bar{\beta}}$ are the normalised states orthogonal to $\ket{\alpha}$ and $\ket\beta$ respectively. Let $D^\dag_\psi= u_0 \ket 0 \bra 0 + u_1 \ket 1 \bra 1$ and $D^\dag_\phi= v_0 \ket 0 \bra 0 + v_1 \ket 1 \bra 1$, then
\begin{align}
\ket\psi &= u_1\ket{\bar{\alpha}\bar{\beta}}-u_0\ket{\bar{\beta}\bar{\alpha}}, \\
\ket\phi &=v_1\ket{\bar{\alpha}\bar{\beta}}-v_0\ket{\bar{\beta}\bar{\alpha}}.
\end{align}
These two states are required to be linearly independent, i.e. the vectors of eigenvalues $(u_0,u_1)$ and $(v_0,v_1)$ must be linearly independent. Hence the FF Hamiltonian with ground space $G_L = \mathrm{span}\{\ket{\alpha}^{\otimes L},\ket{\beta}^{\otimes L}\}$ has rank-2 dimer terms $h_{i,i+1}$ with range equal to $\mathrm{span}(\ket\psi,\ket\phi)$. Any such operator can be written as
\begin{align}
h_{i,i+1} = [A \ket{\eta}\bra{\eta} + B\ket{\bar\eta}\bra{\bar\eta}]_{i,i+1},
\end{align}
where $A,B >0$ and $\{\ket{\eta}, \ket{\bar\eta}\}$ is an orthonormal basis for $\mathrm{span}(\ket\psi,\ket\phi)$.

{We can construct one orthonormal basis by choosing coefficients $(u_0,u_1)$ and $(v_0,v_1)$ such that $\ket\psi$ and $\ket\phi$ are orthonormal (the choice is the same for any $\ket\alpha$, $\ket\beta$ and so this procedure is independent of $G_L$).} Fixing $u_0 = u_1 =[\sqrt 2\det Q^*]^{-1}$, then $\ket\psi \rightarrow \ket{\Psi} :=  \frac{1}{\sqrt 2}(\ket{01}-\ket{10})$, the singlet state. 
Let $\ket{\Phi}$ denote the $\ket\phi$ with $(v_0,v_1)$ chosen so that $\braket{\Psi}{\phi} = 0$. This yields $v_1 = -v_0$. Therefore 
$\ket\phi \rightarrow \ket{\Phi} := \frac{1}{\sqrt{N}} (\ket{\bar{\alpha}\bar{\beta}}+\ket{\bar{\beta}\bar{\alpha}})$, where $N$ ensures  $\braket{\Phi}{\Phi} = 1$.

An arbitrary orthonormal basis $\{\ket{\eta}, \ket{\bar\eta}\}$ for $\mathrm{span}(\ket\psi,\ket\phi)$ can be achieved by applying a two-qubit unitary transformation $U$ to $\{\ket\Psi,\ket\Phi\}$ that only rotates in this two-dimensional subspace. Such a unitary $U$ is of the form
\begin{align} 
U = \exp[-i\theta \ve n \cdot \ve \sigma],
\end{align}
where $\ve n \in \mathbb{R}^3$ is a unit vector, $\theta \in \mathbb{R}$, and $\ve \sigma$ is the vector of Pauli-like matrices in the orthonormal basis  $\{\ket\Psi,\ket\Phi\}$, i.e.
\begin{align}\label{sig1}
\sigma_1  &= \ket\Psi\bra\Phi + \ket\Psi \bra\Phi,\\
\label{sig2}
\sigma_2  &= -i\ket\Psi\bra\Phi + i\ket\Psi \bra\Phi,\\
\label{sig3}
\sigma_3  &= \ket\Psi\bra\Psi - \ket\Phi \bra\Phi.
\end{align}
Hence
\begin{align}
h_{i,i+1} &= [U(A \ket{\Psi}\bra{\Psi} + B\ket{\Phi}\bra{\Phi})U^\dag]_{i,i+1},\\
\ket{\Psi} &=  \frac{1}{\sqrt 2}(\ket{01}-\ket{10}),\\
\ket{\Phi} &\propto \ket{\bar{\alpha}\bar{\beta}}+\ket{\bar{\beta}\bar{\alpha}},
\end{align}
is the most general dimer term with ground space $G_L = \{\ket\alpha^{\otimes L},\ket\beta^{\otimes L}\}$.

Now, in order for the Hamiltonian $H = \sum_i h_{i,i+1}$ to be a valid Majorana zero mode Hamiltonian, it must preserve fermionic parity. In the spin picture this condition is $[H,Z^{\otimes L}] = 0$, 
which is equivalent to $[h_{i,i+1},Z_i Z_{i+1}] = 0$ for all $i$, and this can only be satisfied if
\begin{align}\label{parity_Psi}
Z \otimes Z U\ket{\Psi} &= e^{i\theta_\Psi}U\ket{\Psi},\\
\label{parity_Phi}
Z \otimes Z U\ket{\Phi} &= e^{i\theta_\Phi}U\ket{\Phi},
\end{align}
where $\theta_\Psi,\theta_\Phi \in \mathbb{R}$, and we have dropped the site label $i$. These are eigenvalue equations for the parity operator $Z \otimes Z$. The eigenvalues of $Z \otimes Z$ are $+ 1$ and $-1$, and the corresponding eigenvectors are even and odd parity respectively. In qubit language, even (odd) parity states are linear combinations of computational basis states with an even (odd) number of 1s. Hence, for two qubits, the even parity sector is spanned by $\{\ket{00},\ket{11}\}$, and odd parity is spanned by $\{\ket{01},\ket{10}\}$.

The states $\ket{\Psi}$ and $\ket{\Phi}$ are respectively antisymmetric and symmetric under exchange of the two spins. Already the singlet $\ket{\Psi} = \frac{1}{\sqrt 2}(\ket{01}-\ket{10})$ has parity $-1$. The state $\ket{\Phi}$ is some state in the symmetric subspace, and it can always be written as a linear combination of definite parity states
\begin{align}
\ket{\Phi} = w_+\ket{\Phi_{+1}} + w_-\ket{\Phi_{-1}}, \quad |w_+|^2+|w_-|^2=1,
\end{align}
where
\begin{align}
\ket{\Phi_{+1}} = \cos\left(\frac\theta 2\right)\ket{00} + e^{i\omega}\sin\left(\frac\theta 2\right)\ket{11}, \quad \theta \in [0,\pi], \quad \omega \in [0,2\pi),
\end{align}
and
\begin{align}
\ket{\Phi_{-1}} = \frac{1}{\sqrt 2}(\ket{01} + \ket{10}).
\end{align} 
The unitary $U$ then produces linear combinations of $\ket{\Psi}$ and $\ket{\Phi}$. It is not difficult to see that there are only two categories of $\ket{\Psi}$ and $\ket{\Phi}$ that yield parity preserving dimer terms $h_{i,i+1}$. As $\ket{\Psi}$ has definite parity $-1$, the only way it can combine with $\ket{\Phi}$ to produce new states of definite parity is if $\ket{\Phi}=\ket{\Phi_{-1}}$. There is another independent solution, that is the pair $\ket{\Psi}$ and $\ket{\Phi}=\ket{\Phi_{+1}}$, and only trivial unitaries $U$ are allowed (identity and swaps). We summarise as: 
\bigskip

\emph{Category A}
\begin{align}
\ket{\Psi} &= \frac{1}{\sqrt 2}(\ket{01}-\ket{10}),\\
\ket{\Phi} &=\ket{\Phi_{-1}} = \frac{1}{\sqrt 2}(\ket{01} + \ket{10}),
\end{align}
with any unitary $U$ rotating in span$(\ket\Psi,\ket\Phi)$ allowed. Note that this is equivalent to defining
\begin{align}
\ket{\Psi} &= \ket{01},\\
\ket{\Phi} &=\ket{10}.
\end{align}
and allowing any unitary $U$ rotating in span$(\ket{01},\ket{10})$.
\bigskip

\emph{Category B}
\begin{align}
\ket{\Psi} &= \frac{1}{\sqrt 2}(\ket{01}-\ket{10}),\\
\ket{\Phi} &=\ket{\Phi_{+1}} = \cos\left(\frac\theta 2\right)\ket{00} + e^{i\omega}\sin\left(\frac\theta 2\right)\ket{11}, \quad \theta \in [0,\pi], \quad \omega \in [0,2\pi).
\end{align}
\bigskip

Note that we can always find a local basis such that the parity-conserving Hamiltonian is also real (and therefore time-reversal symmetric). To see why, consider first category A.
Vectors of the form
\begin{align}
U\ket\Psi = \cos\frac a 2 \ket{01} + e^{ib} \sin\frac a 2 \ket{10},\\
U\ket{\Phi_{-1}} = \sin\frac a 2 \ket{01} - e^{ib} \cos\frac a 2 \ket{10}
\end{align}
will appear in $h_{i,i+1}$. Applying the product of qubit unitaries $\mathcal S=\otimes_{k=1}^L S_k$, where $S_k = \ket 0\bra 0 + e^{i(k-1)b}\ket1\bra1$, commutes with the parity operator $Z^{\otimes L}$, preserves translational invariance of the Hamiltonian and removes the phase $e^{ib}$. Without loss of generality then, we can always consider $U$ to be real. 

Now consider Category B. The singlet $\ket\Psi$ has the property that $V \otimes V \ket\Psi = \det V\ket\Psi$ for any unitary $V \in U(2)$. Let $V =\ket 0 \bra 0 + e^{-i\omega/2}\ket 1 \bra 1,$ then applying $V \otimes V$ to the singlet $\ket\Psi$ and to $\ket{\Phi_{+1}}$ gets rid of the phase $e^{i\omega}$. Hence, since  $[V,Z]=0$, applying $V^{\otimes L}$ to $H$ results in a real and translationally invariant Hamiltonian that is parity conserving. 
For this reason, without loss of generality, we need only consider real $\ket{\Phi_{+1}} = \cos\left(\frac\theta 2\right)\ket{00} + \sin\left(\frac\theta 2\right)\ket{11}$.

Given these results, we now need to determine the ground space vectors $\ket\alpha$ and $\ket\beta$. Recall that $\ket{\Phi} = \frac{1}{\sqrt{N}} (\ket{\bar{\alpha}\bar{\beta}}+\ket{\bar{\beta}\bar{\alpha}})$.  We can parametrise the qubit states in the standard way:
\begin{align}
\ket{\alpha} &= \cos\frac u 2 \ket 0 + e^{iv}\sin\frac u 2\ket 1, \quad \ket{\bar\alpha} = \sin\frac u 2 \ket 0 - e^{iv}\cos\frac u 2\ket 1,\\
\ket{\beta} &= \cos\frac x 2 \ket 0 + e^{iy}\sin\frac x 2\ket 1, \quad \ket{\bar\beta} = \sin\frac x 2 \ket 0 - e^{iy}\cos\frac x 2\ket 1,
\end{align}
where $u,x \in [0,\pi]$ and $v,y \in [0,2\pi)$. Then
\begin{align}
\ket{\Phi} =  \frac{2}{\sqrt{N}} &(\sin \frac u 2 \sin \frac x 2 \ket{00} + e^{i(y+v)} \cos \frac u 2 \cos \frac x 2 \ket{11} \\
& -  \frac 1 {\sqrt 2} [e^{iy}\sin \frac u 2 \cos \frac x 2  + e^{iv}\sin \frac x 2 \cos \frac u 2]\ket{\Phi_{-1}}).
\end{align}

Category A: $\ket{\Phi} = \ket{\Phi_{-1}}$. There are two ways this can be achieved:
\begin{align}
(a) \quad\sin \frac u 2 &= \cos \frac x 2 = 0 \Rightarrow u = 0, \quad \mathrm{and}\quad x = \pi \\
(b) \quad\sin \frac x 2 &= \cos \frac u 2 = 0 \Rightarrow x = 0, \quad \mathrm{and}\quad u = \pi.
\end{align}
In case (a) we find $\ket\alpha = \ket 0$ and $\ket\beta =  e^{iv}\ket 1$, and in case (b) $\ket\alpha = e^{iy}\ket 1$ and $\ket\beta =  \ket 0$. Since global phases are irrelevant, these two cases give the same solution.

Category B: $\ket{\Phi} = \ket{\Phi_{+1}} \propto \cos\left(\frac\theta 2\right)\ket{00} + \sin\left(\frac\theta 2\right)\ket{11}$. This occurs when $y+v = 2\pi N$ and
\begin{align}
e^{iy}\sin \frac u 2 \cos \frac x 2  + e^{iv}\sin \frac x 2 \cos \frac u 2 = 0.
\end{align}
Rearranging yields
\begin{align}
\tan \frac u 2 e^{2iy} = - \tan \frac x 2.
\end{align}
Resolving into real and imaginary parts:
\begin{align}\label{real_tan}
\tan \frac u 2 \cos 2y &= - \tan \frac x 2,\\
\label{imag_tan}
\tan \frac u 2 \sin 2y &= 0.
\end{align}
The last equation has a solution with $u=0$ and $y \in [0,2\pi)$. This implies $x=0$, and we find $\ket\alpha = \ket\beta = \ket 0$, however this violates the requirement that $\ket\alpha$ and $\ket\beta$ are independent vectors. The only other solution to Eq.\ \eqref{imag_tan} is $y = n\pi/2$ for $n\in \{0,1,2,3\}$ since $y \in [0,2\pi)$.

If $n \in \{0,2\}$, then Eq.\ \eqref{real_tan} implies that $\frac u 2 = - \frac x 2 + m \pi \Rightarrow u = -x + 2m\pi$, where $m$ is an integer, but as $u,x \in [0,\pi]$, then $u=-x$. The vectors in this case are:
\begin{align}
\ket\alpha = \cos\frac u 2 \ket 0 + \sin\frac u 2\ket 1, \quad \ket{\beta} &= \cos\frac u 2 \ket 0 - \sin\frac u 2\ket 1, \quad \mathrm{for} \quad n=0,\\
\ket\alpha = \cos\frac u 2 \ket 0 - \sin\frac u 2\ket 1, \quad \ket{\beta} &= \cos\frac u 2 \ket 0 + \sin\frac u 2\ket 1, \quad \mathrm{for} \quad n=2,
\end{align}
so both these values of $n$ give the same solution. Since we must have independent $\ket\alpha$ and $\ket\beta$, this restricts $u \in (0,\pi)$. The state $\ket{\Phi_{+1}} \rightarrow\ket{\Phi^0_{+1}} = \frac{2}{\sqrt{N}} ( - \sin^2 \frac u 2\ket{00} +  \cos^2 \frac u 2 \ket{11})$, and the dimer Hamiltonian here is
\begin{align}
h^0_{i,i+1} = A\ket\Psi\bra\Psi + B\ket{\Phi^0_{+1}}\bra{\Phi^0_{+1}}.
\end{align}

If $n \in \{1,3\}$, then Eq.\ \eqref{real_tan} implies that $\frac u 2 = \frac x 2 + m \pi \Rightarrow u = x + 2m\pi$, where $m$ is an integer, but as $u,x \in [0,\pi]$, then $u=x$. The vectors in this case are:
\begin{align}
\ket\alpha = \cos\frac u 2 \ket 0 - i\sin\frac u 2\ket 1, \quad \ket{\beta} &= \cos\frac u 2 \ket 0 + i\sin\frac u 2\ket 1, \quad \mathrm{for} \quad n=1,\\
\ket\alpha = \cos\frac u 2 \ket 0 + i\sin\frac u 2\ket 1, \quad \ket{\beta} &= \cos\frac u 2 \ket 0 - i\sin\frac u 2\ket 1, \quad \mathrm{for} \quad n=3,
\end{align}
so both these values of $n$ give the same solution. Since we must have independent $\ket\alpha$ and $\ket\beta$, this restricts $u \in (0,\pi)$. The state $\ket{\Phi_{+1}}\rightarrow\ket{\Phi^1_{+1}}= \frac{2}{\sqrt{N}} ( \sin^2 \frac u 2\ket{00} +  \cos^2 \frac u 2 \ket{11})$, and the dimer Hamiltonian here is
\begin{align}h^1_{i,i+1} = A\ket\Psi\bra\Psi + B\ket{\Phi^1_{+1}}\bra{\Phi^1_{+1}}.\end{align}
Notice that $\ket{\Phi^1_{+1}} = S \otimes S \ket{\Phi^0_{+1}}$, where $S = i\ket 0 \bra 0 + \ket 1 \bra 1$, and since $\ket\Psi$ is the singlet, we find
\begin{align}h^1_{i,i+1} = S\otimes S h^0_{i,i+1} S^\dag\otimes S^\dag.\end{align}
So since the Hamiltonian with these dimer terms are related by $S^{\otimes L}$ and $[S,Z]=0$, the topological properties of their resulting Hamiltonians will be the same.

Hence we may summarise as follows. Up to a product of qubit unitaries $\mathcal U = U_1\otimes U_2 \otimes \dots \otimes U_L$ that commutes with parity $Z^{\otimes L}$ and preserves translation invariance, the parity symmetric $L$-qubit FF Hamiltonian $H = \sum_i h_{i,i+1}$ with a two-dimensional ground space $G_L = \mathrm{span}\{\ket{\alpha}^{\otimes L},\ket{\beta}^{\otimes L}\}$ has dimer terms
\begin{align} h_{i,i+1} = [U (A\ket\Psi\bra\Psi + B \ket\Phi\bra\Phi)U^\dag]_{i,i+1}, \quad A,B > 0,\end{align}
that split up into two Types:
\begin{itemize}
\item \emph{Type 1}
\begin{align}
\ket\Psi &= \frac{1}{\sqrt 2}(\ket{01}-\ket{10}),\\
\ket\Phi & = \frac{1}{\sqrt{N}} ( \sin^2 \frac u 2\ket{00} +  \cos^2 \frac u 2 \ket{11}), \quad N = \sqrt{\sin^4 \frac u 2+ \cos^4 \frac u 2}
\end{align}
where $u \in (0,\pi)$ and $U =\I$. The ground space states are $\ket\alpha = \cos\frac u 2 \ket 0 + i\sin\frac u 2\ket 1$ and $\ket\beta =\cos\frac u 2 \ket 0 - i\sin\frac u 2\ket 1 = \ket{\alpha^*} = Z\ket\alpha$.

\item \emph{Type 2}
\begin{align}
\ket\Psi &= \cos\frac\gamma 2\ket{01}+ \sin\frac\gamma 2\ket{10},\\
\ket\Phi &=\sin\frac\gamma 2\ket{01}-\cos\frac\gamma 2\ket{10},
\end{align}
with $\gamma \in (0,\pi)$.
The ground space states are $\ket\alpha = \ket 0$ and $\ket\beta = \ket 1$.
\end{itemize}

\section{\label{app_JW}The Jordan-Wigner transformation}

Here we discuss the conversion of a qubit Hamiltonian into a spinless fermion Hamiltonian using a Jordan Wigner transformation. Let $c^\dag_j$ and $c_j$ be the (spinless) fermion creation and annihilation operators, respectively, at site $j$ with $j = 1,\dots,L$. They obey the standard anticommutation relations: $\{c_j,c_k\}=0$,  $\{c^\dag_j,c^\dag_k\}=0$ and  $\{c^\dag_j,c_k\}=\delta_{jk}$. Then the Jordan-Wigner transformation between Pauli spin operators and fermion operators is
\begin{align}
X_j &= \left[\bigotimes_{k=1}^{j-1}Z_k\right] (c^\dag_j+c_j ),\\
Y_j &=i \left[\bigotimes_{k=1}^{j-1}Z_k\right] (c^\dag_j-c_j),\\
Z_j &= \I - 2n_j,
\end{align}
where $n_j = c^\dag_jc_j$ is the fermionic number operator.

In order to have the correct symmetries, the spin Hamiltonians that can support Majorana zero edge modes will only contain terms like $X_j  X_{j+1}, Y_{j} Y_{j+1}, Z_j  Z_{j+1}$ and 1-local $Z_j$. In terms of fermionic operators, the $Z$ terms are straightforward, and we find
\begin{align}
X_j  X_{j+1} &= c_j^\dag c_{j+1}+c_{j+1}^\dag c_j- c_jc_{j+1}-c^\dag_{j+1}c_j^\dag ,\\
Y_j  Y_{j+1} &= c_j^\dag c_{j+1}+c_{j+1}^\dag c_j+c_jc_{j+1}+c^\dag_{j+1}c_j^\dag .
\end{align}
Inserting these expressions into the qubit Hamiltonians immediately yields the results in the main text.

\section{\label{app_MPS}MPS form of the states in the case (i) ground space $G_L$}

The matrix product state formulation of an $L$-qudit state $\ket\psi$ is
\begin{align}
\ket\psi=\sum_{i_1,\dots,i_L} \tr(A^{[i_1]}\dots A^{[i_L]})\ket{i_1\dots i_L},
\end{align}
where $i_k \in \{0,\dots,d-1\}$ for all $k = 1\dots L$ and the $A^{[i_k]}$ are $D \times D$ matrices. Then
\begin{align}
\notag &\bigotimes_{k=1}^{L} F^{[k]} \ket\psi=\sum_{i_1,\dots,i_L} \tr(A^{[i_1]}\dots A^{[i_L]})F^{[1]}\ket{i_1}\dots F^{[L]}\ket{i_L}\\
\notag &=\sum_{j_1,\dots,j_L}\sum_{i_1,\dots,i_L} \tr(A^{[i_1]}\dots A^{[i_L]})f_{j_1i_1}^{[1]}\ket{j_1}\dots f_{j_Li_L}^{[L]}\ket{j_L}\\
&=\sum_{j_1,\dots,j_L}\tr(B^{[j_1]}\dots B^{[j_L]})\ket{j_1\dots j_L}
\end{align}
where \begin{align}F^{[k]} = \sum_{j_k l_k} f_{j_kl_k}^{[k]} \ket{j_k}\bra{l_k},\end{align} and 
\begin{align}B^{[j_k]} = \sum_{i_k}f_{j_ki_k}^{[k]}A^{[i_k]}, \end{align} for $k = 1,\dots, L$.

The MPS matrices for product state $\ket 0^{\otimes L}$ are $C^{[i_k]}=\delta_{i0}$, and for $\ket 1^{\otimes L}$ are $\bar C^{[i_k]}=\delta_{i1}$, for any $k=1,\dots,L$. Therefore $\ket\alpha^{\otimes L} = (Q\ket 0)^{\otimes L}$ has MPS matrices $A^{[i_k]}=q_{i0}$, while $\ket\beta^{\otimes L} = (Q\ket 1)^{\otimes L}$ has $ B^{[i_k]}=q_{i1}$, where $q_{ij}$ are the entries of the matrix $Q$.

The MPS form of a superposition can be achieved with block-diagonal matrices: 
\begin{align}
u\ket{\psi} + v\ket{\phi}&=v\sum_{i_1,\dots,i_L} \tr(A^{[i_1]}\dots A^{[i_L]})\ket{i_1\dots i_L}\\
&+v\sum_{i_1,\dots,i_L} \tr(B^{[i_1]}\dots B^{[i_L]})\ket{i_1\dots i_L}\\
&= \sum_{i_1,\dots,i_L} \tr(W^{[i_1]}\dots W^{[i_L]})\ket{i_1\dots i_L}
\end{align}
where
\begin{align}
W^{[i_1]} = \begin{pmatrix}
uA^{[i_1]} & 0 \\
0 & vB^{[i_1]}
\end{pmatrix},
\end{align}
and
\begin{align}
W^{[i_k]} = \begin{pmatrix}
A^{[i_k]} & 0 \\
0 & B^{[i_k]}
\end{pmatrix},
\end{align}
for all $k = 2,\dots,L$. Of course there are many choices for where to absorb the coefficients $u$ and $v$, here we attach them to the first spin. This construction holds for MPS with periodic boundary conditions, however, we are interested in open chains (so that MZMs can exist at the ends), and this requires that $\dim(W^{[i_1]}) = \dim(W^{[i_L]})= 1$, in other words, $W^{[i_1]} = (uA^{[i_1]},vB^{[i_1]})$ and $ W^{[i_L]} = (A^{[i_L]},B^{[i_L]})^T$ are row and column vectors respectively (in which case the trace over all the $W$s is redundant).

\bibliography{references}

\end{document}